\documentclass[10pt,a4paper]{article}
\usepackage[utf8]{inputenc}
\usepackage{authblk}
\usepackage{amsmath}
\usepackage{natbib}
\usepackage{amsthm}
\usepackage{hyperref}
\usepackage{algorithm}
\usepackage[noend]{algpseudocode}
\usepackage{tikz}
\usetikzlibrary{shapes}
\usetikzlibrary{arrows}

\newtheorem{definition}{Definition}
\newtheorem{theorem}{Theorem}
\newtheorem{lemma}{Lemma}
\newtheorem{example}{Example}
\newtheorem{corollary}{Corollary}

\newcommand{\Vars}{\ensuremath{\mathcal{V}}}
\newcommand{\Fsig}{\ensuremath{\mathcal{F}}}
\newcommand{\talg}[2]{\ensuremath{\mathcal{T}(#1,#2)}}
\newcommand{\genOfst}[1]{\ensuremath{\prec_{#1}}}
\newcommand{\genOf}[1]{\ensuremath{\preceq_{#1}}}
\newcommand{\approxOf}[1]{\ensuremath{\approx_{#1}}}
\newcommand{\Egen}[3]{\mathcal{G}_{#1}(#2,#3)}
\newcommand{\mcsg}{\mathit{mcsg}}

\author{ David M. Cerna\thanks{Email: dcerna@cs.cas.cz,\ Orchid:0000-0002-6352-603X, \ Funding: Czech Science Foundation Grant 22-06414L and Cost
Action CA20111 EuroProofNet.} }
\date{%
    Czech Academy of Sciences Institute of Computer Science\\[2ex]%
    \today
}
\title{A Note On Square-free Sequences and Anti-unification Type}

\begin{document}

\maketitle

\begin{abstract}
\textbf{Error: Peer-review process exposed an error in Theorem~\ref{thm:three} that, unfourtunately, is not repairable. Idempotent semigroups are always finite. See~\cite{Green_Rees_1952,DBLP:conf/ki/SiekmannS81a} for details}
\textit{Anti-unification} is a fundamental operation used for inductive inference. It is abstractly defined as a process deriving from a set of symbolic expressions a new symbolic expression possessing certain
commonalities shared between its members. We consider anti-unification over term algebras where some function symbols are interpreted as \textit{associative-idempotent} ($f(x,f(y,z)) = f(f(x,y),z)$ and $f(x,x)=x$, respectively) and show that there exists generalization problems for which a \textit{minimal complete set of solutions} does not exist (\textit{Nullary}), that is every complete set must contain comparable elements with respect to the \textit{generality relation}. In contrast to earlier techniques for showing the nullarity of a generalization problem, we exploit combinatorial properties of complete sets of solutions to show that comparable elements are not avoidable. We show that every complete set of solutions contains an infinite chain of comparable generalizations whose structure is isomorphic to a subsequence of an infinite \textit{square-free} sequence over three symbols.   
\end{abstract}

\section{Introduction}
\label{sec:intro}
\textit{Anti-unification} is a fundamental operation used for inductive inference. It is abstractly defined as a process deriving from a set of symbolic expressions a new symbolic expression possessing certain commonalities shared between its members. In this work, we consider anti-unification over term algebras where some function symbols are interpreted as \textit{associative-idempotent}, i.e., over the equational theory $\{f(x,f(y,z)) = f(f(x,y),z)\ , \ f(x,x)=x\}$. We use results from the theory of square-free sequences~\citep{thue1906} to show that for certain anti-unification problems, a minimal complete set of solutions does not exist, i.e., the theory is \textit{Nullary}. Concerning unification, this theory was considered independently by F. Baader~\cite{10.1007/BF02328451} and Schmidt-Schauss~\cite{10.1007/BF02328450} and shown to be \textit{Nullary}. The argument presented in this paper exploits the existence of square-free infinite sequences over a ternary alphabet. Our construction is theory-specific, but it is likely generalizable to a method for deducing the anti-unification type of related theories. We leave this to future work.   
\section{Preliminaries}
\label{sec:prelim}
Let  \Vars\ be a countable set of variables and \Fsig\ a set of function symbols with a fixed arity. The set of terms derived from \Fsig\ and \Vars\ is 
denoted by $\talg{\Fsig}{\Vars}$, whose members are constructed using the grammar  $t ::= x 
\mid f(t_1,\dots,t_n)$, where $x\in\Vars$ and $f\in \Fsig$ with arity $n\geq 0$. When 
$n=0$, $f$ is called a {\it constant}. 
A term $t$ is \emph{ground} if no variables occur in $t$. The {\em size} of a term is defined inductively as: $|x|=1$, and 
$\textstyle |f(t_1,\ldots,t_n)| = 1 + \sum_{i=1}^{n} |t_i|$. A {\it substitution} is a function $\sigma: \Vars \to \talg{\Fsig}{\Vars}$ such that $\sigma(x)\neq  x$ for only finitely many variables. Substitutions are extended to terms as usual. We use the postfix notation for substitution application to terms and write $t\sigma$ instead of $\sigma(t)$.

\subsection{Equational Anti-unification}
In this work, we focus on equational anti-unification. Thus, we refrain from presenting syntactic variants of the concepts discussed below. For such details, we refer to the recent survey on the topic~\cite{DBLP:conf/ijcai/CernaK23}.

\begin{definition}[Equational theory \cite{DBLP:books/daglib/0092409}]
\label{def:equtheo}
Let $E$ be a set of equational axioms. The relation $\approxOf{E} =: \{ (s,t) \in \talg{\Fsig}{\Vars} \times \talg{\Fsig}{\Vars} \; | \;  E\models s\approx t\}$ is  called the \emph{equational theory} induced by $E$.  
\end{definition}

\begin{definition}[$E$-generalization] 
\label{def:egen}
Let $s,t\in \talg{\Fsig}{\Vars}$ and $E$ a set of equational axioms. Then \emph{$t$ is more general then $s$ over $E$}, denoted $s\genOf{E} t$, if there exists a substitution $\sigma$ such that $s\sigma\approxOf{E} t$. Furthermore, we refer to a term $r\in \talg{\Fsig}{\Vars}$ as an \emph{$E$-generalization} of $s$ and $t$ if $r\genOf{E}  s$ and $r\genOf{E} t$. The set of all $E$-generalizations of $s$ and $t$ is denoted as $\Egen{E}{s}{t}$. By $\prec_E$, we denote the strict relation induced by $\genOf{E}$.
\end{definition}

\begin{definition}[Minimal complete set]
\label{def:mcsg}
Let $s,t\in \talg{\Fsig}{\Vars}$. Then the \emph{minimal complete set of $E$-generalizations of  $s$ and $t$}, denoted as $\mcsg_E(s,t)$,  is a subset of $\Egen{E}{s}{t}$ satisfying:
\begin{itemize}
\item[] \hspace{-2em}\textbf{Completeness:} For each $r\in \Egen{E}{s}{t}$ there exists $r'\in \mcsg_E(s,t)$ such that $r\genOf{E} r'$.
\item[] \hspace{-2em}\textbf{Minimality:} If $r,r'\in \mcsg_E(s,t)$ and $r\genOf{E} r'$, then $r=r'$.
\end{itemize}
\end{definition}

\begin{definition}[Anti-unification type]
\label{def:typemcsg}
The anti-unification type of an equational theory $E$ may have one of the following forms: \begin{itemize}
\item[] \hspace{-2em}\textbf{Unitary:} $\mcsg_{E}(s,t)$ exists for all  $s,t\in \talg{\Fsig}{\Vars}$ and is  singleton. 
\item[] \hspace{-2em}\textbf{Finitary:} $\mcsg_{E}(s,t)$ exists and is finite for all  $s,t\in \talg{\Fsig}{\Vars}$, and there exist  $s',t'\in  \talg{\Fsig}{\Vars}$ for which  $1 <$ 

\hspace*{2.1em} $|\mcsg_{E}(s',t')| < \infty$. 
\item[] \hspace{-2em}\textbf{Infinitary:} $\mcsg_{E}(s,t)$ exists for all  $s,t\in \talg{\Fsig}{\Vars}$, and there exist $s',t'\in  \talg{\Fsig}{\Vars}$ such that  $\mcsg_{E}(s',t')$  is infinite. 
\item[] \hspace{-2em}\textbf{Nullary:} for some  $s,t\in   \talg{\Fsig}{\Vars}$, $\mcsg_{E}(s,t)$ does not exist. 
\end{itemize}
\end{definition}

\begin{example}  
Syntactic AU  is \textit{unitary}~\cite{Plotkin70,Reynolds70}, AU over associative (A) and commutative (C)  theories is \textit{finitary}~\cite{DBLP:journals/iandc/AlpuenteEEM14},  AU over idempotent theories is \textit{infinitary}~\cite{DBLP:journals/tocl/CernaK20}, and AU with multiple unital equations is \textit{nullary}~\cite{DBLP:conf/fscd/CernaK20}.
\end{example}

We focus on associative-idempotent theories, i.e., some function symbols $f\in \mathcal{F}$ are interpreted over the equational axioms $f(x,f(y,z)) = f(f(x,y),z)$ and $f(x,x)=x$. This theory will be denoted $\textit{AI}$. We say that a term $t$ is in \emph{ idempotent flat normal form} if (i) $t\in \mathcal{V}$, (ii) $t= g(t_1,...,t_n)$ where $n\geq 0$, $g$ is a free function (constant) symbol and each $t_i$, $1 \leq i \leq n$, is in idempotent flat normal form, or (iii) $t= f(t_1,...,t_n)$ where $f$ is an associative-idempotent symbol, $n\geq 2$, for $1\leq i\leq n$, $t_i$ is a term whose head symbol is not $f$, and the $t_i$'s are mutually distinct. Furthermore, given a term $t= f(t_1,\ldots,t_n)$ where $f$ is associative-idempotent function symbol, t is in idempotent flat normal form, and $n>1$, we define  $t\vert_i = t_i$ and $t\hspace*{-.2em}\parallel_j = f(t_j,\ldots,t_n)$ where  $1\leq i\leq n$, $1\leq j\leq n$, and if $j= n$, then $t\hspace*{-.2em}\parallel_j = t_n$.

\section{Nullarity of AI Anti-unification}
\label{subsec1}

\noindent To show nullarity of associative-idempotent anti-unification we consider $E$-generalizations of the terms $h(a,b)$ and $h(b,a)$ where $h$ is an AI function symbol. First, we show that particular syntactically restricted terms generalizes this problem. 
\begin{definition}[simple words] 
\label{def:simwor}
Let $C\subset \talg{\Fsig}{\Vars}$ be a non-empty set of constants and $x\in \mathcal{V}$. Then the \emph{$(x,C)$-simple words}, denoted $\mathcal{W}(x,C)$, is the set of idempotent flat normal form terms constructed from the  following set:
\[\{ h(x,s_1,\ldots,s_n,x) \mid   n\geq 0 \wedge \forall i ( 1\leq i\leq n \Rightarrow  s_i\in \{x\}\cup C)\}.\]
\end{definition}
\noindent The following theorem shows that $\mathcal{W}(x,C)\subseteq \mathcal{G}_{AI}(h(a,b),h(b,a))$.
\begin{lemma}
\label{thm:one}
Let $a,b\in \talg{\Fsig}{\Vars}$ be constants, $x\in \mathcal{V}$, and $t\in \mathcal{W}(x,\{a,b\})$. Then $t\sigma \approxOf{AI} s$ where $x\sigma = s$ and $s\in \{h(a,b),h(b,a)\}$.

\end{lemma}
\begin{proof}
Without loss of generality, we assume that $x\sigma = h(a,b)$. Then $t\sigma = h(a,b,r_1,\ldots, r_m,a,b)$ where $m\geq 0$ and for all $1\leq i \leq m$, $r_i\in \{a,b\}$.  
We continue by well-founded induction on $m$. Let $t\sigma = h(a,b,r_1,\ldots,r_{m+1},a,b)$. If either $r_1=b$ or $r_{m+1}=a$, then $t\sigma \approxOf{AI} h(a,b,r_1',\ldots,r_{m}',a,b)$ and by the induction hypothesis the theorem holds. If $r_1=a$ and $r_{m+1}=b$,  Then $t\sigma = h(a,b,a,r_2,\ldots,r_{m},b,a,b)$ and no matter the choice for $r_2$ and $r_m$, $t\sigma \approxOf{AI} h(a,b,r_1',\ldots,r_{k}',a,b)$, where $k<m+1$,  thus completing the proof.
\end{proof}

\noindent Lemma~\ref{thm:one} entails the following simple corollary.
\begin{corollary}
\label{cor:one}
Let $a,b\in \talg{\Fsig}{\Vars}$ be constants, $x\in \mathcal{V}$. Then $\mathcal{W}(x,\{a,b\})\subseteq \mathcal{G}_{\textit{AI}}(h(a,b),h(b,a)).$
\end{corollary}

There may be generalizations of $h(a,b)$ and $h(b,a)$ that are not contained in $\mathcal{W}(x,\{a,b\})$, i.g. $h(x,t,y)$.  However, any term of the form $h(t_1,t_2)$ or $h(t_2,t_1)$ where $t_1\in \{a,b\}$ is not a generalization of $h(a,b)$ and $h(b,a)$. The following lemma addresses this observation.

\begin{lemma}
\label{thm:two}
Let $a,b\in \talg{\Fsig}{\Vars}$ be constants and $t\in  \talg{\Fsig}{\Vars} $ be a term of the form $h(t_1,t_2)$ or $h(t_2,t_1)$ where $t_1\in \{a,b\}$. Then $t\not \in\mathcal{G}_{AI}(h(a,b),h(b,a))$.
\end{lemma}
\begin{proof}
Consider $t= h(a,t')$ without loss of generality. Let $\sigma$ be a substitution. Observe that either $t\sigma \approxOf{AI} a$ or $t\sigma \approxOf{AI} h(a,s_1,\ldots,s_n)$ in idempotent flat normal form.  In either case, $t$ does not generalize $h(b,a)$. 
\end{proof}
Observe that terms in $\mathcal{W}(x,\{a,b\})$ form words over an alphabet of three symbols. Exploiting the result by  \textit{Axel Thue}~\citep{thue1906} (See \href{https://oeis.org/A005678}{oeis.org/A005678}) that, over ternary alphabets, infinite square-free words are constructible, we build an infinite sequence of generalizations strictly ordered by $\genOfst{\textit{AI}}$. In particular, we use the infinite word found in the OEIS as sequence \href{https://oeis.org/A005678}{A005678} to construct the infinite sequence of generalizations. Using Algorithm~\ref{alg:genconstruct} we construct the set $\mathcal{W}_{th}(x,\{a,b\})= \{\textrm{sqGen}(n)\mid n\geq 1\}$. Observe that $\mathcal{W}_{th}(x,\{a,b\})\subset \mathcal{W}(x,\{a,b\})\subseteq \mathcal{G}_{AI}(h(a,b),h(b,a))$.

    \begin{algorithm}
        \caption{Construction of square-free generalizations}\label{euclid}
        \begin{algorithmic}[1]
            \Procedure{sqGen}{$n$}    \Comment{We assume that $n\geq1$}
			 \State $g\leftarrow h(a,b,a,x)$			         
            \For{\{$i=1$ ; $i<n$ ; $i\hspace*{.1em}\raisebox{.1em}{${\scriptstyle +\hspace*{-.01em}+}$}$\}} 
				\If{$g\vert_{i+1}= a$}
					\State $g\leftarrow h(g,a,b,a,x)$  \Comment{We flatten $\textbf{g}$  in $h(\textbf{g},a,b,a,x)$}			         
				\EndIf            
            	\If{$g\vert_{i+1}= b$}
					\State $g\leftarrow h(g,b,a,b,x)$	\Comment{We flatten $\textbf{g}$  in $h(\textbf{g},b,a,b,x)$}		         
				\EndIf 
				\If{$g\vert_{i+1}= x$ and $g\vert_{i}= a$ }
					\State $g\leftarrow h(g,b,x,a,x)$		\Comment{We flatten $\textbf{g}$  in $h(\textbf{g},b,x,a,x)$}	         
				\EndIf 
				\If{$g\vert_{i+1}= x$ and $g\vert_{i}= b$ }
					\State $g\leftarrow h(g,a,x,b,x)$	\Comment{We flatten $\textbf{g}$  in $h(\textbf{g},a,x,b,x)$}		         
				\EndIf 

            \EndFor
            \State \textbf{return} $g\hspace{-.2em}\parallel_{4}$
            \EndProcedure
        \end{algorithmic}
        \label{alg:genconstruct}
    \end{algorithm}

\begin{example}
The following terms are contained in $\mathcal{W}_{th}(x,\{a,b\})$:
\begin{align*}
\textrm{sqGen}(1)=&\ x & \textrm{sqGen}(3)=&\ h(x,b,a,b,x,a,b,a,x)\\
\textrm{sqGen}(2)=&\ h(x,b,a,b,x) & \textrm{sqGen}(4)=&\ h(x,b,a,b,x,a,b,a,x,b,x,a,x)
\end{align*}
\end{example}

\noindent Now we prove the nullarity of associative-idempotent anti-unification:
\begin{theorem}
\label{thm:three}
Let $g_1\in \mathcal{W}(x,\{a,b\})$, $g_2\in \mathcal{W}_{th}(x,\{a,b\})$ where $|g_2| > |g_1|$,  $g_3\approxOf{AI} g_1\sigma\in  \mathcal{W}(x,\{a,b\})$, and $x\sigma = g_2$.  Then $|g_3|\geq |g_2|$ and $g_1\genOfst{AI} g_3$.
\end{theorem}
\begin{proof}$g_1\sigma = h(g_2,s_1,\ldots,s_n,g_2)$ where $n\geq 0$ and for $1\leq i\leq n$, $s_i\in\{a,b\}$ or $s_i=g_2$. If $n=0$, then $g_1\sigma \approxOf{AI} g_2$, i.e. $|g_1\sigma| \geq |g_2|$. We show that it is always the case that $|g_1\sigma| \geq |g_2|$ by well-founded induction on $n$. Let $g_1\sigma = h(s_0,s_1,\ldots,s_{n+1},s_{n+2})$ where $s_0=s_{n+2}=g_2$. If $g_1\sigma$ is in idempotent flat normal form, then $|g_1\sigma| \geq |g_2|$  trivially holds. If not, then there exists $0\leq i< n+2$ and $1\leq k\leq \frac{n+2}{2}$ such that the sequence $s_i,\ldots,s_{i+(k-1)}$ is equivalent to $s_i,\ldots,s_{i+(2k-1)}$. Thus, $g_1\sigma  \approxOf{AI} h(g_2,s_1,\ldots,s_{m},g_2)$ where $m\leq n$ and by the induction hypothesis $|g_1\sigma| \geq |g_2|$. Thus, for any term $t$, where $t\approxOf{AI} g_1\sigma$, $|t| \geq |g_2|$.  This implies $|g_3| \geq |g_2|$. By a similar argument,  there does not exist a substitution $\sigma$ such that $|g_2\sigma| < |g_2|$. Furthermore, $g_3\approxOf{AI} t \in  \mathcal{W}(x,\{a,b\})$, and  $g_1\genOfst{AI} g_3$. 
\end{proof}

\begin{theorem}
Associative-idempotent anti-unification is nullary. 
\end{theorem}
\begin{proof}
We assume that $\mathcal{C}= \mcsg_E(h(a,b),h(b,a))$ exists. Then there exist  $g_1\in \mathcal{C}$ such that $h(x, b, a, b, x)\genOf{AI} g_1$. Furthermore, $g_1 = h(t_1,\ldots,t_n)$ and forall $1\leq i\leq n$, $t_i$ only contains occurances of  variables and the constants $a$ and $b$. By Lemma~\ref{thm:two}, $g_1$ has the form $h(t_1,t_2,t_3)$ where $t_1,t_3\in \Vars$. Furthermore, applying a substitution $\mu$ to $g_1$ that maps all variables to $x$ results in a term $g_3\in \mathcal{W}(x,\{a,b\})$ that generalizes $h(a,b)$ and $h(b,a)$ (Lemma~\ref{thm:one}). Now let $g_2\in \mathcal{W}_{th}(x,\{a,b\})$ be such that $|g_2| > |g_3|$. By Theorem~\ref{thm:three},  $|g_3\sigma|\geq |g_2|$ and  $g_1\genOfst{AI}g_3\sigma$ where $x\sigma = g_2$. Observe that, by completeness, there exists $g_4
\in \mathcal{C}$ such that $g_3\sigma\genOf{AI} g_4$. However, then $g_1\genOfst{AI}g_4$ contradicting the minimality of $\mathcal{C}$, thus completing the proof.
\end{proof}
\noindent\textbf{Acknowledgements:} I would like to thank Temur Kutsia for the comments and suggestions that improved this work.

\bibliographystyle{abbrvnat}
\bibliography{references}

\end{document}